%% file: costpltl_arxiv.tex
\documentclass[a4paper, fleqn]{llncs}

\input{preamble}

\title{Parameterized Linear Temporal Logics Meet Costs:\\ Still not Costlier than LTL\thanks{Supported by the project “TriCS”~(ZI 1516/1-1) of the German Research Foundation (DFG). A preliminary version of this work appeared in the proceedings of GandALF 2015~\cite{Zimmermann15gandalf}.}}

\author{Martin Zimmermann}

\institute{Reactive Systems Group, Saarland University, 66123 Saarbrücken, Germany\\
 \email{\{zimmermann\}@react.uni-saarland.de}}

\begin{document}

\maketitle

\begin{abstract}
\input{abstract}
\end{abstract}

\section{Introduction}
\label{sec_intro}
\input{intro}

\section{Parametric LTL with Costs}
\label{sec_defs}
\input{defs}

\section{The Alternating-Color Technique for Costs}
\label{sec_altcolor}
\input{altcolor}

\section{Model Checking}
\label{sec_mc}
\input{mc}

\section{Infinite Games}
\label{sec_games}
\input{games}

\section{Parametric LDL with Costs}
\label{sec_ldl}
\input{cpldl}

\section{Multiple Cost Functions}
\label{sec_mult}
\input{multcost}

\section{Optimization Problems}
\label{sec_opt}
\input{optimization}

\section{Conclusion}
\label{sec_conc}
\input{conc}

\bibliographystyle{plain}
\bibliography{biblio}

\end{document}

%% file: preamble.tex
\usepackage[utf8]{inputenc}

\usepackage{amsmath}
\usepackage{amssymb}

\usepackage{wrapfig}

\usepackage{tikz}
\usetikzlibrary{automata}

\newcommand{\myquot}[1]{``#1''}

\newcommand{\coloneq}{\mathop{:=}}

\newcommand{\cceq}{\mathop{::=}}

\renewcommand{\epsilon}{\varepsilon}

\newcommand{\pow}[1]{2^{#1}}
\newcommand{\nats}{\mathbb{N}}
\newcommand{\card}[1]{|#1|}
\newcommand{\size}[1]{\card{#1}}
\newcommand{\set}[1]{\{#1\}}

\newcommand{\inc}{\kappa}

\newcommand{\ttrue}{\texttt{tt}}
\newcommand{\ffalse}{\texttt{ff}}
\newcommand{\F}{\mathbf{F}}
\newcommand{\G}{\mathbf{G}}
\newcommand{\U}{\mathbf{U}}
\newcommand{\X}{\mathbf{X}}
\newcommand{\R}{\mathbf{R}}

\newcommand{\Var}{\mathcal{V}} 

\newcommand{\cl}{\mathrm{cl}}

\newcommand{\var}{\mathrm{var}} 
\newcommand{\varG}{\mathrm{var}_{\mathbf{G}}}
\newcommand{\varF}{\mathrm{var}_{\mathbf{F}}}

\newcommand{\halfthinspace}{{\kern .08333em}}

\newcommand{\conc}{\,;}

\newcommand{\ddiamond}[1]{\langle\/ #1 \/\rangle\,}
\newcommand{\bbox}[1]{[\halfthinspace#1\halfthinspace]\,}

\newcommand{\ddiamondle}[2]{{\langle\/ #1 \/\rangle}_{\!\le #2}\,}
\newcommand{\bboxle}[2]{{[\halfthinspace#1\halfthinspace]}_{\le #2}\,}

\newcommand{\Rexp}{\mathcal{R}}

\newcommand{\rel}[1]{\mathrm{rel}(#1)}

\newcommand{\ltl}{\text{LTL}}

\newcommand{\prompt}{\text{PROMPT}$\textendash$\ltl}

\newcommand{\pltl}{\text{PLTL}}

\newcommand{\pltlg}{\pltl_{{\mathbf{G}}}}

\newcommand{\pltlc}{\text{cPLTL}}
\newcommand{\pltlcf}{\pltlc_{{\mathbf{F}}}}
\newcommand{\pltlcg}{\pltlc_{{\mathbf{G}}}}

\newcommand{\ldl}{\text{LDL}}

\newcommand{\pldl}{\text{PLDL}}
\newcommand{\pldldiamond}{\text{PLDL}_{\Diamond}}

\newcommand{\pldlc}{\text{cPLDL}}
\newcommand{\pldlcdiamond}{\pldlc_{\Diamond}}
\newcommand{\pldlcbox}{\pldlc_{\Box}}

\newcommand{\mpltlc}{\text{mult-}\pltlc}
\newcommand{\mpldlc}{\text{mult-}\pldlc}
\newcommand{\mpltlcf}{\text{mult-}\pltlcf}
\newcommand{\mpltlcg}{\text{mult-}\pltlcg}

\newcommand{\vldl}{\text{VLDL}}
\newcommand{\vpmutl}{VP\text{-}\mu TL}

\newcommand{\pspace}{\textsc{{PSpace}}}
\newcommand{\exptime}{\textsc{{ExpTime}}}
\newcommand{\twoexp}{\textsc{{2ExpTime}}}
\newcommand{\threeexp}{\textsc{{3ExpTime}}}

\newcommand{\sys}{\mathcal{S}}
\newcommand{\trace}{\mathrm{tr}}
\newcommand{\cst}{\text{cst}}

\newcommand{\aut}{\mathfrak{A}}

\newcommand{\arena}{\mathcal{A}}
\newcommand{\game}{\mathcal{G}}

\newcommand{\mem}{\mathcal{M}}

\newcommand{\update}{\mathrm{upd}}
\newcommand{\nextmove}{\mathrm{nxt}}

%% file: abstract.tex
We continue the investigation of parameterized extensions of Linear Temporal Logic (LTL) that retain the attractive algorithmic properties of LTL: a polynomial space model checking algorithm and a doubly-exponential time algorithm for solving games. Alur et al.\ and Kupferman et al.\ showed that this is the case for Parametric LTL (PLTL) and PROMPT-LTL respectively, which have temporal operators equipped with variables that bound their scope in time. Later, this was also shown to be true for Parametric LDL (PLDL), which extends PLTL to be able to express all $\omega$-regular properties. 

Here, we generalize PLTL to systems with costs, i.e., we do not bound the scope of operators in time, but bound the scope in terms of the cost accumulated during time. Again, we show that model checking and solving games for specifications in  PLTL with costs is not harder than the corresponding problems for LTL. Finally, we discuss PLDL with costs and extensions to multiple cost functions. 

%% file: intro.tex
Parameterized linear temporal logics address a serious shortcoming of Linear-temporal Logic ($\ltl$)~\cite{Pnueli77}: $\ltl$ is not able to express timing constraints, e.g., while $\G(q \rightarrow \F p)$ expresses that every request~$q$ is eventually answered by a response~$p$, the waiting time between requests and responses might diverge. This is typically not the desired behavior, but cannot be ruled out by $\ltl$.

To overcome this shortcoming, Alur et al.\ introduced parameterized $\ltl$ ($\pltl$)~\cite{AlurEtessamiLaTorrePeled01}, which extends $\ltl$ with parameterized operators of the form $\F_{\le x}$ and $\G_{\le y}$, where $x$ and $y$ are variables. The formula~$\G(q \rightarrow \F_{\le x} p)$ expresses  that every request is answered within an arbitrary, but fixed number of steps~$\alpha(x)$. Here, $\alpha$ is a variable valuation, a mapping of variables to natural numbers. Typically, one is interested in whether a $\pltl$ formula is satisfied with respect to some variable valuation. For example, the model checking problem asks whether a given transition system satisfies a given $\pltl$ specification~$\varphi$ with respect to some $\alpha$, i.e., whether every path satisfies $\varphi$ with respect to $\alpha$. Similarly, solving infinite games amounts to determining whether there is an $\alpha$ such that Player~$0$ has a strategy such that every play that is consistent with the strategy satisfies the winning condition with respect to $\alpha$. Alur et al.\ showed that the $\pltl$ model checking problem is $\pspace$-complete. Kupferman et al.\ later considered $\prompt$~\cite{KupfermanPitermanVardi09}, which can be seen as the fragment of $\pltl$ without the parameterized always operator, and showed that $\prompt$ model checking is still $\pspace$-complete and that $\prompt$ realizability, an abstract notion of infinite game, is $\twoexp$-complete. While the results of Alur et al. relied on involved pumping arguments, the results of Kupferman et al.\ where all based on the so-called alternating-color technique, which basically allows to reduce $\prompt$ to $\ltl$. Furthermore, the result on realizability was extended to infinite games on graphs~\cite{Zimmermann13}, again using the alternating-color technique.

Another serious shortcoming of $\ltl$ (and its parameterized variants) is their expressiveness: $\ltl$ is equi-expressive to first-order logic with order~\cite{Kamp68} and thus not as expressive as $\omega$-regular expressions. This shortcoming was addressed by a long line of temporal logics~\cite{GiacomoVardi13,LeuckerSanchez07,Vardi11,VardiWolper94,Wolper1983} with regular expressions, finite automata, or grammar operators to obtain the full expressivity of the $\omega$-regular languages. One of these logics is Linear Dynamic Logic ($\ldl$), which has temporal operators~$\ddiamond{r}$ and $\bbox{r}$, where $r$ is a regular expression. For example, the formula~$\bbox{r_0}(q \rightarrow \ddiamond{r_1} p)$ holds in a word~$w$, if every request at a position~$n$ such that $w_0 \cdots w_{n-1}$ matches $r_0$, there is a position~$n' \ge n$ such that $p$ holds at $n'$ and $w_n \cdots w_{n'-1}$ matches $r_1$. Intuitively, the diamond operator corresponds to the eventuality of $\ltl$, but is guarded by a regular expression. Dually, the box operator is a guarded always. 
Although $\ldl$ is more expressive than $\ltl$, its algorithmic properties are similar: model checking is $\pspace$-complete and solving games is $\twoexp$-complete~\cite{Vardi11}. 

There are temporal logics whose expressiveness goes even beyond the $\omega$-regular languages to capture properties of recursive programs, which are typically $\omega$-contextfree. The visibly $\omega$-contextfree languages~\cite{AlurM04} are an important class of languages located between the $\omega$-regular ones and the $\omega$-contextfree ones that enjoys desirable closure properties, which make it suitable to be employed in verification. Temporal logics that capture this class are visibly $\ltl$~\cite{bozzelli14b}, the fixed-point logic~$\vpmutl$~\cite{bozzelli07}, and visibly~$\ldl$ ($\vldl$)~\cite{WeinertZimmermann15}. The logic~visibly~$\ltl$ enhances $\ltl$ with visibly rational expressions~\cite{bozzelli14a}, and $\vpmutl$ extends the linear-time $\mu$-calculus~\cite{Vardi88} with non-local modalities. Finally, $\vldl$ has the same temporal operators as $\ldl$, but allows to use visibly pushdown automata instead of regular expressions as guards. For all these logics, model checking is $\exptime$-complete, i.e., (under standard complexity theoretic assumptions) harder than the model checking problem for $\ltl$. Furthermore, solving games with $\vldl$ winning conditions is $\threeexp$-complete, again harder than solving $\ltl$ games.  Thus, going beyond the $\omega$-regular languages does increase the complexity of these problems at last. 

All these logics tackle one shortcoming of $\ltl$, but not both simultaneously. This was achieved for the first time by adding parameterized operators to $\ldl$. The logic, called parameterized $\ldl$ ($\pldl$), has additional operators $\ddiamondle{r}{x}$ and $\bboxle{r}{y}$ with the expected semantics: the variables bound the scope of the operator. And even for this logic, which has parameters and is more expressive than $\ltl$, model checking is still $\pspace$-complete and solving games is $\twoexp$-complete~\cite{FaymonvilleZimmermann14}. Again, these problems were solved by an application of the alternating-color technique. One has to overcome some technicalities, but the general proof technique is the same as for $\prompt$. 

The decision problems for the parameterized logics mentioned above are boundedness problems, e.g., one asks for an upper bound on the waiting times between requests and responses in case of the formula~$\G(q \rightarrow \F_{\le x} p)$. Recently, more general boundedness problems in logics and automata received a lot of attention to obtain decidable quantitative extensions of monadic second-order logic and better synthesis algorithms. In general, boundedness problems are undecidable for automata with counters, but become decidable if the acceptance conditions can refer to boundedness properties of the counters, but the transition relation has no access to counter values. Recent advances include logics and automata with bounds~\cite{Bojanczyk04,BojanczykColcombet06}, satisfiability algorithms for
these logics~\cite{Bojanczyk11,Bojanczyk14,BojanczykTorunczyk12,Boom11}, and regular
cost-functions~\cite{Colcombet09}. However, these formalisms, while very expressive, are intractable and thus not suitable for verification and synthesis. 
Thus, less expressive formalisms were studied that appear more suitable for practical applications, e.g., finitary parity~\cite{ChatterjeeHenzingerHorn09}, parity with costs~\cite{FijalkowZimmermann14}, energy-parity~\cite{ChatterjeeDoyen10}, mean-payoff-parity~\cite{ChatterjeeHenzingerJurdzinski05}, consumption games~\cite{BCKN12}, and the use of weighted automata for specifying quantitative properties~\cite{BCHJ09,CernyChatterjeeHenzingerRadhakrishnaSingh11}. In particular, the parity condition with cost is defined in graphs whose edges are weighted by natural numbers (interpreted as costs) and requires the existence of a bound~$b$ such that almost every occurrence of an odd color is followed by an occurrence of a larger even color such that the cost between these positions is at most $b$. Although strictly stronger than the classical parity condition, solving parity games with costs is as hard as solving parity games~~\cite{FijalkowZimmermann14,MogaveroMS13}. 

We investigate parameterized temporal logics in a weighted setting similar to the one of parity conditions with costs: our graphs are equipped with cost-functions that label the edges with natural numbers and parameterized operators are now evaluated with respect to cost instead of time, i.e., the parameters bound the accumulated cost instead of the elapsed time. Thus, the formula~$\G( q \rightarrow \F_{\le x}p)$ requires that every request~$q$ is answered with cost at most $\alpha(x)$. We show the following results about $\pltl$ with costs ($\pltlc$):

First, we refine the alternating-color technique to the cost-setting, which requires to tackle some technical problems induced by the fact that accumulated cost, unlike time, does not increase in every step, e.g., if an edge with cost zero is traversed. In particular, infinite paths with finite cost have to be taken care of appropriately. 

Second, we show that Kupferman et al.'s proofs based on the alternating-color technique can be adapted to the cost-setting as well. For model checking, we again obtain $\pspace$-completeness while solving games is still $\twoexp$-complete.

Third, we consider $\pldl$ with costs ($\pldlc$), which is defined as expected: the diamond and the box operator may be equipped with parameters bounding their scope. Again, the complexity does not increase: model checking is $\pspace$-complete while solving games is $\twoexp$-complete. 

Fourth, we generalize both logics to a setting with multiple cost-functions. Now, the parameterized temporal operators have another parameter that determines the cost-function under which they are evaluated. Even these extensions do not increase complexity: model checking is again $\pspace$-complete while solving games is still $\twoexp$-complete. 

Fifth, we also study the optimization variant of the  model checking and the game problem for these logics: here, one is interested in finding the \emph{optimal} variable valuation for which a given transition system satisfies the specification. For example, for the request-response condition one is interested in minimizing the waiting times between requests and responses. For $\pltlc$ and $\pldlc$, we show that the model checking optimization problem can be solved in polynomial space while the optimization problem for infinite games can be solved in triply-exponential time. These results are similar to the ones obtained for $
\pltl$~\cite{AlurEtessamiLaTorrePeled01,Zimmermann13}. In particular, the exponential gap between the decision and the optimization variant of solving infinite games exists already for $\pltl$. Whether this gap can be closed is an open problem. A first step towards this direction was made by giving an approximation algorithm for this problem with doubly-exponential running time~\cite{TentrupWeinertZimmermann15}.

The paper is structured as follows: in Section~\ref{sec_defs}, we introduce $\pltlc$ and discuss basic properties. Then, in Section~\ref{sec_altcolor}, we extend the alternating-color technique to the setting with costs, which we apply in Section~\ref{sec_mc} to the model checking problem and in Section~\ref{sec_games} to solve infinite games. In Section~\ref{sec_ldl}, we extend these results to $\pldlc$ and to multiple cost-functions. Finally, in Section~\ref{sec_opt}, we investigate model checking and game-solving as optimization problems. 

%% file: defs.tex
Let $\Var$ be an infinite set of variables and let $P$ be a set of atomic propositions. The formulas of $\pltlc$ are given by the grammar
\begin{equation*}\varphi \cceq p \mid \neg p \mid \varphi \wedge \varphi \mid \varphi \vee
\varphi
  \mid \X \varphi \mid \varphi \U \varphi \mid \varphi \R \varphi \mid
  \mathbf{F}_{ \le z } \varphi \mid \mathbf{G}_{ \le z} \varphi
,\end{equation*}
where $p \in P$ and $z \in \Var$. We use the derived operators $\ttrue \coloneq p
\vee \neg p$ and $\ffalse \coloneq p \wedge \neg p$ for some fixed $p \in P$, $\F
\varphi \coloneq \ttrue \U \varphi$, and $\G \varphi \coloneq \ffalse \R \varphi$. Furthermore, we use $p \rightarrow \varphi$ and $\neg p \rightarrow \varphi$ as shorthand for $\neg p \vee \varphi$ and $p \vee \varphi$, respectively. Additional derived operators are introduced on page~\pageref{page_derivedops}.

The set of subformulas of a $\pltlc$ formula~$\varphi$ is denoted by $\cl( \varphi )$
and we define the size of $\varphi$ to be the cardinality of $\cl(\varphi)$.
Furthermore, we define 
\[\varF( \varphi ) = \{ z\in \Var \mid \F_{\le z} \psi \in
\cl( \varphi) \}\] to be the set of variables parameterizing eventually operators in
$\varphi$, \[\varG( \varphi ) = \{ z\in \Var \mid \G_{\le z} \psi \in \cl( \varphi)  \} \]
to be the set of variables parameterizing always operators in $\varphi$. Furthermore,
$\var( \varphi ) = \varF( \varphi ) \cup \varG( \varphi )$ denotes the set of all variables appearing in $\varphi$. 

$\pltlc$ is evaluated on so-called cost-traces (traces for short) of the form 
\[w = w_0\, c_0\, w_1\, c_1\, w_2\, c_2\, \cdots \in \left( \pow{P} \cdot\ \nats \right)^{ \omega },\] which encode the evolution of the system in terms of the atomic propositions that hold true in each time instance, and the cost of changing the system state. The cost of the trace~$w$ is defined as $\cst(w) = \sum_{j \ge 0}c_j$, which might be infinite. A finite cost-trace is required to begin and end with an element of $\pow{P}$. The cost~$\cst(w)$ of a finite cost-trace~$w = w_0 c_0 w_1 c_1 \cdots c_{n-1} w_n$ is defined as $\cst(w) = \sum_{j=0}^{n-1}c_j$. 

Furthermore, we require the existence of a distinguished atomic proposition~$\inc$ such that all cost-traces satisfy $c_j >0 $ if and only if $\inc \in w_{j+1}$, i.e., $\inc$ indicates that the last step had non-zero cost. We use the proposition~$\inc$ to reason about costs: for example, we are able to express whether a trace has cost~zero and whether a trace has cost~$\infty$. In the following, we will ensure that all our systems only allow traces that satisfy this assumption. 

Also, to evaluate formulas we need to  instantiate the variables parameterizing the temporal operators. To this end, we define a variable valuation to be a
mapping~$\alpha\colon \Var \rightarrow \nats$. 
Now, we can define the model
relation between a cost-trace~$w = w_0\, c_0\, w_1\, c_1\, w_2\, c_2\, \cdots $, a
position~$n$ of $w$, a variable valuation~$\alpha$, and a~$\pltlc$ formula as
follows:
\begin{itemize}
\item $(w,n,\alpha)\models p$ if and only if  $p\in w_n$,

\item $(w,n,\alpha)\models\neg p$ if and only if  $p\notin
w_n$,

\item $(w,n,\alpha)\models\varphi\wedge\psi$ if and only if
$(w,n,\alpha)\models\varphi$ and
$(w,n,\alpha)\models\psi$,

\item $(w,n,\alpha)\models\varphi\vee\psi$ if and only if
$(w,n,\alpha)\models\varphi$ or $(w,n,\alpha)\models\psi$,

\item $(w,n,\alpha)\models\X\varphi$ if and only if  $(w,n+1,\alpha)\models\varphi$,

\item $(w,n,\alpha)\models\varphi\U\psi$ if and only if there exists a $j\ge 0$
such that
$(w,n+j,\alpha)\models\psi$ and $(w,n+k,\alpha)\models\varphi$ for every $k$ in the
range $0\le k< j $,

\item $(w,n,\alpha)\models\varphi\R\psi$ if and only if  for every $j\ge 0$: either
$(w,n+j,\alpha)\models\psi$
or there exists a $k$ in the range $0\le k < j$ such that
$(w,n+k,\alpha)\models\varphi$,

\item $(w,n,\alpha)\models\F_{\le z}\varphi$ if and only if there exists a $j \ge 0$ with \newline$\cst(w_{n} c_n \cdots c_{n+j-1} w_{n+j}) \le \alpha(z)$ such that $(w,n+j,\alpha)\models\varphi$, and

\item $(w,n,\alpha)\models\G_{\le z}\varphi$ if and only if for every $j \ge 0$ with \newline$\cst(w_{n} c_n \cdots c_{n+j-1} w_{n+j}) \le \alpha(z)$: $(w,n+j,\alpha)\models\varphi$.

\end{itemize}
Note that we recover the semantics of $\pltl$ as the special case where every~$c_n$ is equal to one. 

For the sake of brevity, we write $(w,\alpha) \models \varphi$ instead of
$(w,0,\alpha) \models \varphi$ and say that $w$ is a model of $\varphi$ with
respect to $\alpha$. For variable-free formulas, we even drop the $\alpha$ and write $w \models \varphi$.

As usual for parameterized temporal logics, the use of variables has to be
restricted: bounding eventually and always operators by the same variable leads
to an undecidable satisfiability problem~\cite{AlurEtessamiLaTorrePeled01}.

\begin{definition}
\label{def_wellformedformula}
A $\pltlc$ formula~$\varphi$ is well-formed, if $\varF( \varphi ) \cap \varG( \varphi ) =
\emptyset$.
\end{definition}

In the following, we only consider well-formed formulas and omit the qualifier~\myquot{well-formed}. Also, we will denote
variables in $\varF( \varphi )$ by $x$ and variables in $\varG( \varphi )$ by $y$, if the formula~$\varphi$ is clear from context.

We consider the following fragments of $\pltlc$. Let $\varphi$ be a $\pltlc$ formula:
\begin{itemize}
\item $\varphi$ is an $\ltl$ formula, if $\var( \varphi ) = \emptyset$.

\item $\varphi$ is a $\pltlcf$ formula, if $\varG(\varphi) = \emptyset$.

\item $\varphi$ is a $\pltlcg$ formula, if $\varF(\varphi) = \emptyset$.

\end{itemize}
Every $\ltl$, $\pltlcf$, and every $\pltlcg$ formula is well-formed by
 definition.

\begin{example}\hfill
\begin{enumerate}
	\item The formula~$\G (q \rightarrow \F_{\le x}p)$ is satisfied with respect to $\alpha$, if every request (a position where $q$ holds) is followed by a response (a position where $p$ holds) such that the cost of the infix between the request and the response is at most $\alpha(x)$.
	\item The (max-) parity condition with costs~\cite{FijalkowZimmermann14} can be expressed\footnote{\label{footnote_costparity}Note that the bound in the parity condition with costs may depend on the trace while one typically uses global bounds for $\pltlc$ (see, e.g., Section~\ref{sec_mc} and Section~\ref{sec_games}). However, for games in finite arenas (and thus also for model checking) these two variants coincide~\cite{FijalkowZimmermann14}.} in $\pltlc$ via
\[
\F\G  \left( \bigwedge\nolimits_{c \in \set{1, 3, \ldots, d-1}} \left(c \rightarrow  \F{}_{\le x}\bigvee\nolimits_{c' \in \set{c+1, c+3, \ldots, d}} c' 
\right)\right),
\]
where $d$ is the maximal color, which we assume w.l.o.g.\ to be even. However, the Streett condition with costs~\cite{FijalkowZimmermann14} cannot be expressed in $\pltlc$, as it is defined with respect to multiple cost functions, one for each Streett pair. We extend $\pltlc$ to multiple cost functions in Section~\ref{sec_mult}.

\end{enumerate}
\end{example}

As for $\pltl$, one can also parameterize the until and the release operator and also consider bounds of the form~\myquot{$>\!z$}. However, this does not increase expressiveness of the logic. Formally, we define \label{page_derivedops}
\begin{itemize}

\item $(w,n,\alpha)\models\varphi \U_{\le z}\psi$ if and only if there exists a $j \ge 0$ with \newline$\cst(w_{n} c_n \cdots c_{n+j-1} w_{n+j}) \le \alpha(z)$ such that $(w,n+j,\alpha)\models\psi$ and $(w,n+k,\alpha)\models\varphi$ for every $k$ in the range $0\le k< j $, 

\item $(w,n,\alpha)\models\varphi\R_{\le z}\psi$ if and only if for every $j \ge 0$ with \newline$\cst(w_{n} c_n \cdots c_{n+j-1} w_{n+j}) \le \alpha(z)$: $(w,n+j,\alpha)\models\psi$ or there exists a $k$ in the range $0\le k < j$ such that
$(w,n+k,\alpha)\models\varphi$,
\item $(w,n,\alpha)\models\F_{> z}\varphi$ if and only if there exists a $j \ge 0$ with \newline$\cst(w_{n} c_n \cdots c_{n+j-1} w_{n+j}) > \alpha(z)$ such that $(w,n+j,\alpha)\models\varphi$, and

\item $(w,n,\alpha)\models\G_{> z}\varphi$ if and only if for every $j \ge 0$ with $\newline\cst(w_{n} c_n \cdots c_{n+j-1} w_{n+j}) > \alpha(z)$ satisfies $(w,n+j,\alpha)\models\varphi$.

\item $(w,n,\alpha)\models\varphi \U_{> z}\psi$ if and only if there exists a $j \ge 0$ with \newline$\cst(w_{n} c_n \cdots c_{n+j-1} w_{n+j}) > \alpha(z)$ such that $(w,n+j,\alpha)\models\psi$ and $(w,n+k,\alpha)\models\varphi$ for every $k$ in the range $0\le k< j $, and

\item $(w,n,\alpha)\models\varphi\R_{> z}\psi$ if and only if for every $j \ge 0$ with \newline $\cst(w_{n} c_n \cdots c_{n+j-1} w_{n+j}) > \alpha(z)$: $(w,n+j,\alpha)\models\psi$ or there exists a $k$ in the range $0\le k < j$ such that
$(w,n+k,\alpha)\models\varphi$.

\end{itemize}
Let $\varphi \equiv \psi$ denote equivalence of the formulas~$\varphi$ and $\psi$, i.e., for every $w$, every $n$, and every $\alpha$, we have $(w, n, \alpha) \models \varphi$ if and only if $(w, n, \alpha) \models \psi$. Then, we have the following equivalences (which also restrict the use of variables as defined in Definition~\ref{def_wellformedformula}):\medskip

\begin{minipage}[b]{.4\textwidth}
	\begin{itemize}
	\item $\varphi \U_{\le x} \psi \equiv \varphi \U \psi \wedge \F_{\le x} \psi$
\item 	$\varphi \R_{\le y} \psi \equiv \varphi \R \psi \vee \G_{\le y} \psi$
\item 	$\F_{> y} \varphi \equiv \G_{\le y}\F\X(\inc \wedge \F\varphi)$
	\end{itemize}
	\end{minipage}
\begin{minipage}[b]{.5\textwidth}
	\begin{itemize}
\item	$\G_{> x} \varphi \equiv \F_{\le x}\G\X(\neg \inc \vee \G\varphi)$
\item	$\varphi \U_{> y} \psi \equiv \G_{\le y}(\varphi \wedge \F\X(\inc \wedge \varphi\U\psi))$
\item $	\varphi \R_{> x} \psi \equiv \F_{\le x}(\varphi \vee   \G\X(\neg \inc \vee \varphi\R\psi))$

	\end{itemize}
	\end{minipage}\medskip

Note that we defined $\pltlc$ formulas to be in negation normal form.
Nevertheless, a negation can be pushed to the atomic propositions using the
duality of the operators. Thus, we can define the negation of a $\pltlc$
formula.

\begin{lemma}
\label{lemma_pltlnegation}
For every $\pltlc$ formula~$\varphi$ there exists an efficiently constructible
$\pltlc$ formula~$\neg \varphi$ s.t.\
\begin{enumerate}
\item $(w,n,\alpha)\models \varphi$ if and only if $(w,n,
\alpha) \not\models \neg \varphi$ for every $w$, every $n$, and every $\alpha$,
\item $\card{\neg \varphi} =  \card{\varphi}$.
\item If $\varphi$ is well-formed, then so is $\neg \varphi$.
\item If $\varphi$ is an $\ltl$ formula, then so is $\neg \varphi$.
\item If $\varphi$ is a $\pltlcf$ formula, then $\neg \varphi$ is a $\pltlcg$ formula
and vice versa.
\end{enumerate}
\end{lemma}

\begin{proof}
We construct $\neg \varphi$ by induction over the construction of $\varphi$ using the dualities of the operators:\medskip

\begin{minipage}[b]{.4\textwidth}
	\begin{itemize}
		\item $\neg (p) = \neg p$
		\item $\neg (\varphi \wedge \psi) = (\neg \varphi) \vee (\neg \psi)$
		\item $\neg (\varphi \U \psi) = \neg \varphi \R \neg \psi$
		\item $\neg (\F_{\le x}\varphi ) = \G_{\le x} \neg \varphi$

	\end{itemize}
	\end{minipage}
\begin{minipage}[b]{.4\textwidth}
	\begin{itemize}
		\item $\neg (\neg p) = p$
		\item $\neg (\varphi \vee \psi) = (\neg \varphi) \wedge (\neg \psi)$
		\item $\neg (\varphi \R \psi) = \neg \varphi \U \neg \psi$
 		\item $\neg (\G_{\le y}\varphi ) = \F_{\le y} \neg \varphi$

	\end{itemize}
	\end{minipage}\medskip

The latter four claims of Lemma~\ref{lemma_pltlnegation} follow from the definition of $\neg \varphi$ while the first one can be shown by a straightforward induction over $\varphi$'s construction.
\qed\end{proof}

Another important property of parameterized logics is monotonicity: increasing (decreasing) the values for parameterized eventuality operators (parameterized always operators) preserves satisfaction.

\begin{lemma}
\label{lemma_monotonicity}
Let $\varphi$ be a $\pltlc$ formula and let $\alpha$ and
$\beta$ be variable valuations satisfying $\alpha ( x ) \le \beta ( x)$ for
every $x \in \varF( \varphi)$ and $\alpha ( y ) \ge \beta ( y)$ for every $y \in \varG( \varphi)$. If $(w, \alpha) \models \varphi$, then $(w, \beta) \models
\varphi$.
\end{lemma}

Especially, if we are interested in checking whether a formula is satisfied with respect to some $\alpha$, we can always recursively replace every subformula~$\G_{\le y}\psi$ by 
$\psi \vee \X (\neg \inc \U (\neg \inc \wedge \psi))$, as this is equivalent to $\G_{\le y}\psi$ with respect to every variable valuation mapping $y$ to zero, which is the smallest possible value for $y$. Note that we have to ignore the current truth value of $\inc$, as it indicates the cost of the last transition, not the cost of the next one.

%% file: altcolor.tex
Fix a fresh atomic proposition~$p \notin P$. We say that a cost-trace \[w' = w_0' c_0' w_1' c_1' w_2' c_2' \cdots \in \left( \pow{P \cup \set{p}} \cdot \nats \right)^\omega\] is a coloring of a cost trace \[w = w_0 c_0 w_1 c_1 w_2 c_2 \cdots \in \left( \pow{P}  \cdot \nats \right)^\omega,\] if $w_n' \cap P = w_n$ and $c_n' = c_n$ for every $n$, i.e., $w'$ and $w$ only differ in the truth values of the new proposition~$p$. A position~$n$ is a changepoint of $w'$, if $n=0$ or if the truth value of $p$ in $w_{n-1}'$ and $w_{n}'$ differs. A block of $w'$ is an infix $w_n' c_n' \cdots w_{n+j}'$ of $w'$ such that $n$ and $n+j+1$ are successive changepoints. If a coloring has only finitely many changepoints, then we refer to its suffix starting at the last changepoint as its tail, i.e., the coloring is the concatenation of a finite number of blocks and its tail. 

Let $k \in \nats$. We say that $w'$ is $k$-bounded if every block and its tail (if it has one) has cost at most $k$. Dually, we say that $w'$ is $k$-spaced, if every block has cost at least $k$. Note that we do not have a requirement on the cost of the tail in this case.

Given a $\pltlcf$ formula~$\varphi$, let $\rel{\varphi}$ denote the $\ltl$ formula obtained from $\varphi$ by recursively replacing every subformula~$\F_{\le x} \psi$ by
\[(p \rightarrow p\U(\neg p \U \rel{\psi})) \,\, \wedge \,\, 
(\neg p \rightarrow \neg p\U( p \U \rel{\psi})).
\]
Intuitively, the relativized formula requires $\rel{\psi}$ to be satisfied within at most one changepoint. On bounded and spaced colorings, $\varphi$ and $\rel{\varphi}$ are \myquot{equivalent}.

\begin{lemma}[cp.\ Lemma 2.1 of \cite{KupfermanPitermanVardi09}]
\label{lemma_altcolor}
Let $w$ be a cost-trace and let $\varphi$ be a $\pltlcf$ formula.
\begin{enumerate}
	
\item\label{lemma_altcolor_pltl2ltl}
Let $(w, \alpha) \models \varphi$ for some variable valuation $\alpha$. Then, $w' \models \rel{\varphi}$ for every $(k+1)$-spaced coloring $w'$ of $w$, where $k = \max_{x \in \var(\varphi)}\alpha(x)$.

\item\label{lemma_altcolor_ltl2pltl}
Let $w' \models \rel{\varphi}$ for some $k$-bounded coloring $w'$ of $w$. Then, $(w, \alpha ) \models \varphi$, where $\alpha(x) = 2k$ for every $x$.

\end{enumerate}
\end{lemma}

\begin{proof}

Note that $w$ and its colorings coincide on their cost. Hence, when speaking about the cost of an infix or suffix, we do not have to specify whether we refer to $w$ or to a coloring of $w$.

\ref{lemma_altcolor_pltl2ltl}.) Fix a $(k+1)$-spaced coloring~$w'$ of $w$, where $k = \max_{x \in \var(\varphi)}\alpha(x)$. We show that $(w, n, \alpha) \models \varphi$ implies $(w', n) \models \rel{\varphi}$ by induction over the construction of $\varphi$. 

The only non-trivial case is the one of a parameterized eventuality: thus, assume $(w, n, \alpha) \models \F_{\le x}\psi$, i.e., there is a $j$ with $\cst(w_n c_n \cdots c_{n+j-1} w_{n+j}) \le \alpha(x)$ and $(w, n+j, \alpha) \models \psi$. By induction hypothesis, we have $(w', n+j) \models \rel{\psi}$. As $w'$ is $(k+1)$-spaced, i.e., the cost of each block is at least $k+1$, there is at most one changepoint between (and including) the positions $n$ and $n+j-1$ in $w'$. Hence, $(w', n) \models p\U(\neg p \U \rel{\psi}))$, if $p \in w_{n}'$, and $(w', n) \models  \neg p\U( p \U \rel{\psi}))$ otherwise. Thus, $(w', n) \models \rel{\F_{\le x} \psi}$.

\ref{lemma_altcolor_ltl2pltl}.) Dually, fix a $k$-bounded coloring~$w'$ of $w$ and define the variable valuation~$\alpha$ with $\alpha(x) = 2k$ for every $x$. We show that $(w', n) \models \rel{\varphi}$ implies $(w, n, \alpha) \models \varphi$ by induction over the construction of $\varphi$. 

Again, the only non-trivial case is the one of a parameterized eventuality: thus, let $(w', n) \models \rel{\F_{\le x}\psi}$. We assume $(w', n) \models p$ (the other case is dual). Then, we have 
$ (w', n) \models  p\U(\neg p \U \rel{\psi})$,
i.e., $\rel{\psi}$ is satisfied at some position~$n+j$ such that there is at most one changepoint between (and including) the positions~$n$ and $n+j-1$ in $w'$. As $w'$ is $k$-bounded, this implies that the cost of the infix~$w_n c_n \cdots w_{n+j}$ is bounded by $2k$. Furthermore, applying the induction hypothesis yields $(w,n+j, \alpha) \models \psi$. Hence,  $(w,n, \alpha) \models \F_{\le x}\psi$.
\qed
\end{proof}

%% file: mc.tex
A transition system~$\sys = (S, s_I, E, \ell , \cst)$ consists of a finite directed graph~$(S, E)$, an initial state~$s_I \in S$, a labeling function~$\ell \colon S \rightarrow \pow{P}$, and a cost function $\cst \colon E \rightarrow \nats$. We encode the weights in binary, although the algorithms we present in this section and their running times and space requirements are oblivious to the exact weights. Furthermore, we assume that every state has at least one successor to spare us from dealing with finite paths. Recall our requirement on cost-traces having a distinguished atomic property~$\inc$ indicating the sign of the cost of the previous transition. Thus, we require $\sys$ to satisfy the following property: if $\inc \in \ell(v')$, then $\cst(v,v') >0$ for every edge~$(v,v') \in E$ leading to $v'$. Dually, if $\inc \notin \ell(v')$, then $\cst(v,v') = 0$ for every edge~$(v,v') \in E$.

A path through $\sys$ is a sequence~$\pi = s_0 s_1 s_2 \cdots$ with $s_0 = s_I$ and $(s_n, s_{n+1}) \in E$ for  every $n$. Its cost-trace~$\trace(\pi)$ is defined as
\[
\trace(\pi) = 
\ell(s_0) \cst(s_0, s_1)
\ell(s_1) \cst(s_1, s_2)
\ell(s_2) \cst(s_2, s_3)
\cdots,
\]
which satisfies our assumption on the proposition~$\inc$.

The transition system~$\sys$ satisfies a $\pltlc$ formula~$\varphi$ with respect to a variable valuation~$\alpha$, if the trace of every path through $\sys$ satisfies $\varphi$ with respect to $\alpha$. The $\pltlc$ model checking problem asks, given a transition system~$\sys$ and a $\pltlc$ formula~$\varphi$, whether $\sys$ satisfies $\varphi$ with respect to some~$\alpha$.

\begin{theorem}
\label{thm_mc}
The $\pltlc$ model checking problem is $\pspace$-complete.
\end{theorem}

The proof we give below is a generalization of the one for $\prompt$ by Kupferman et al.~\cite{KupfermanPitermanVardi09}. We begin by showing  $\pspace$-membership. First note that we can restrict ourselves to $\pltlcf$ formulas: given a $\pltlc$ formula~$\varphi$, let $\varphi'$ denote the formula obtained by recursively replacing every subformula $\G_{\le y} \psi$ by $\psi \vee \X (\neg \inc \U (\neg \inc \wedge \psi))$. Due to Lemma~\ref{lemma_monotonicity} and the discussion below it, every transition system~$\sys$ satisfies $\varphi$ with respect to some $\alpha$ if and only if $\sys$ satisfies $\varphi'$ with respect to some~$\alpha'$.

Next, we show how to apply the alternating-color: recall that the classical algorithm for $\ltl$ model checking searches for a fair path, i.e., one that visits infinitely many accepting states, in the product of $\sys$ with a Büchi automaton recognizing the models of the negated specification. If such a path exists, then $\sys$ does not satisfy the specification, as the fair path contains a path~$\pi$ through $\sys$ and an accepting run of the automaton on its trace, i.e., the trace does not satisfy the specification. If there is no such fair path, then $\sys$ satisfies the specification.

For $\pltlc$ we have to find such a path for every $\alpha$ in order to show that $\sys$ does not satisfy the specification with respect to any $\alpha$. To this end, one relativizes the $\pltlcf$ specification as described in Section~\ref{sec_altcolor} and builds an automaton for the negation of the relativized formula in conjunction with a formula that ensures that every ultimately periodic model is both $k$-bounded and $k'$-spaced for some appropriate $k$ and $k'$. Then, we search for a pumpable fair path in the product of the system and the Büchi automaton recognizing the models of the negated specification, which is non-deterministically labeled by $p$. Applying Lemma~\ref{lemma_altcolor} and pumping a fair path through the product appropriately yields a counterexample for every $\alpha$. Thus, model checking is reduced to finding a pumpable fair path. Let us stress again that this algorithm is similar to the one for $\prompt$, we just have to pay attention to some intricacies stemming from the fact that we want to bound the cost, not the waiting time: there might be paths with finite cost, which have to be dealt with appropriately. 

Recall that $p$ is the distinguished atomic proposition used to relativize $\pltlc$ formulas. A colored Büchi graph with costs~$(V, v_I, E, \ell, \cst, F)$ consists of a finite direct graph~$(V, E)$, an initial vertex~$v_I$, a labeling function~$\ell \colon V \rightarrow \pow{\set{p}}$, a cost-function~$\cst \colon E \rightarrow \nats$, and a set~$F \subseteq V$ of accepting vertices. A path $v_0 v_1 v_2 \cdots $ is pumpable, if each of its blocks induced by $p$ contains a vertex repetition such that the cycle formed by the repetition has non-zero cost\footnote{Note that our definition is more involved than the one of Kupferman et al., since we require a cycle with non-zero cost instead of any circle.}. Note that we do not have a requirement on the cost of the tail, if the path has one. The path is fair, if it visits $F$ infinitely often. The pumpable non-emptiness problem asks, given a colored Büchi graph with costs, whether it has an initial pumpable fair path. 

\begin{lemma}
\label{lemma_upnonemptiness}
If a colored Büchi graph with costs has an initial pumpable fair path, then also one of the form~$\pi_0 \pi_1^\omega$ with $\size{\pi_0\pi_1} \in \mathcal{O}(n^2)$, where $n$ is the number of vertices of the graph.
\end{lemma}

\begin{proof}
Let $\pi$ be an arbitrary initial pumpable fair path. First, assume it has only finitely many changepoints. If there are two blocks that start with the same vertex, then we can remove all blocks in between and obtain another initial pumpable fair path. Thus, we can assume that $\pi$ has at most $n$ blocks. Furthermore, the length of each block can be bounded by $\mathcal{O}(n)$ by removing cycles while retaining the state repetition with non-zero cost and at least one accepting vertex (provided the block has one). Now, consider the tail: by removing infixes one can find a cycle of length at most $n$ containing an accepting vertex and a path of length at most $n$ leading from the last changepoint to a vertex on the cycle. Hence, we define $\pi_0$ to be the prefix containing all blocks and the path leading to the cycle and define $\pi_1$ to be the cycle. Then, we have $\size{\pi_0\pi_1} \in \mathcal{O}(n^2)$ and $\pi_0 \pi_1^\omega$ is an initial pumpable fair path.

On the other hand, if $\pi$ contains infinitely many changepoints, then we can remove blocks and shorten other blocks as described above until we have constructed a prefix~$\pi_0 \pi_1$ such that $\pi_0\pi_1^\omega$ has the desired properties. In this case, we can assume that the first position of $\pi_1$ is a changepoint by \myquot{rotating} $\pi_1$ appropriately and appending a suitable prefix of it to $\pi_0$.
\qed\end{proof}

Let $\sys = (S, s_I, E, \ell, \cst)$ be a transition system and let $
\varphi$ be a $\pltlcf$ formula. Furthermore, consider the $\ltl$ formula \[\chi = (\G\F p \wedge \G\F \neg p) \leftrightarrow \G\F \inc,\]
which is satisfied by a cost-trace, if the trace has infinitely many changepoints if and only if\footnote{Here, we use our assumption on $\inc$ indicating the sign of the costs.} it has cost~$\infty$. 
 Now, let $\aut = (Q, \pow{P \cup \set{p}}, q_I, \delta, F)$ be a nondeterministic Büchi automaton recognizing the models of the $\ltl$ formula~$\neg \rel{\varphi} \wedge \chi$, which we can pick such that its number of states is bounded exponentially in $\size{\varphi}$.
 Now, define the colored Büchi graph with costs \[\sys \times \aut = (S \times Q \times \pow{\set{p}}, (s_I, q_I, \emptyset), E', \ell', \cst', F') \] where
\begin{itemize}
	\item $((s,q, C),(s',q', C')) \in E'$ if and only if $(s,s') \in E$ and $q' \in \delta(q, \ell(s)\cup C)$,
	\item $\ell(s,q,C) = C$,
	\item $\cst'((s,q,C),(s',q',C')) = \cst(s,s') $, and
	\item $F' = S \times F \times \pow{\set{p}}$.
\end{itemize}

\begin{lemma}
\label{lemma_mccharacterization}[cp.\ Lemma 4.2 of \cite{KupfermanPitermanVardi09}]
$\sys$ does not satisfy $\varphi$ with respect to any $\alpha$ if and only if $\sys \times \aut$ has an initial pumpable fair path.
\end{lemma}

\begin{proof}
Let $\sys$ not satisfy $\varphi$ with respect to any variable valuation. Fix $k = (\size{S} \cdot \size{Q}+3)\cdot W$, where $W$ is the largest cost in $\sys$, and define the valuation~$\alpha$ by $\alpha(x) = 2k$ for every $x$. As $\sys$ does not satisfy $\varphi$ with respect to $\alpha$, there is a path~$\pi$ through $\sys$ with $(\trace(\pi), \alpha) \not \models \varphi$. Thus, due to Lemma~\ref{lemma_altcolor}.\ref{lemma_altcolor_ltl2pltl}, every  $k$-bounded coloring of $w$ does not satisfy $\rel{\varphi}$.

Now, let $w'$ be a $k$-bounded and $(k-W)$-spaced coloring of $\trace(\pi)$ which starts with $p$ not holding true. Such a coloring can always be constructed, as $W$ is the largest cost appearing in $\sys$. Note that $w'$ satisfies $\chi$ by construction. Thus, we have $w' \models \neg\rel{\varphi} \wedge \chi$, i.e., there is an accepting run~$q_0 q_1 q_2 \cdots $ of $\aut$ on $w'$. Consider the path
\[
(s_0, q_0, w'_0 \cap \set{p})
(s_1, q_1, w'_1 \cap \set{p})
(s_2, q_2, w'_2 \cap \set{p})
\cdots
\]
where $s_0 s_1 s_2 \cdots = \pi$, which is fair by construction. We claim that it is pumpable: consider a block, which is $(k-W)$-spaced. Thus, it contains at least $\size{S} \cdot \size{Q}+2$ many edges with non-zero cost, enough to enforce a vertex repetition with non-zero cost in between. To this end, one takes the sets~$V_j$ of vertices visited between the  $j$-th and the $(j+1)$-th edge with non-zero cost (including the $j$-th edge). This yields $\size{S} \cdot \size{Q}+1$ non-empty sets of vertices of $\sys \times \aut$ that coincide in their third component, as we are within one block. However, there are only $\size{S} \cdot \size{Q}$ many such vertices, which yields the desired repetition. 

Now, consider the converse implication and let $\alpha'$ be an arbitrary variable valuation. We show that $\sys$ does not satisfy $\varphi$ with respect to $\alpha'$. Due to Lemma~\ref{lemma_monotonicity}, it is sufficient to show that $\sys$ does not satisfy $\varphi$ with respect to the valuation~$\alpha$ mapping every variable to $k = \min_{x \in \var(\varphi)}\alpha'(x)$. 

Fix an initial pumpable fair path of $\sys \times \aut$. It has a vertex repetition in every block such that the induced cycle has non-zero cost. We pump each such cycle $k+1$ times to obtain the path
\[
(s_0, q_0, C_0)
(s_1, q_1, C_1)
(s_2, q_2, C_2)
\cdots.
\]
By construction, $\pi = s_0 s_1 s_2 \cdots$ is a path through $\sys$ and 
\[
w' = 
(\ell(s_0) \cup C_0)
(\ell(s_1) \cup C_1)
(\ell(s_2) \cup C_2)
\cdots
\]
is a coloring of its trace~$\trace(\pi)$. Also, $q_0 q_1 q_2 \cdots $ is an accepting run of $\aut$ on $w'$, i.e., $w' \models \neg\rel{\varphi} \wedge \chi$. Lastly, $w'$ is $(k+1)$-spaced by construction, as the cost-function of $\sys \times \aut$ is induced by the one of~$\sys$. 

Assume towards a contradiction that $\sys$ satisfies $\varphi$ with respect to $\alpha$, which implies $(\trace(\pi), \alpha) \models \varphi$. Applying Lemma~\ref{lemma_altcolor}.\ref{lemma_altcolor_pltl2ltl} yields that every $(k+1)$-spaced coloring of $\trace(\pi)$ satisfies $\rel{\varphi}$. However, $w'$ is a $(k+1)$-spaced coloring which satisfies $\neg \rel{\varphi}$, i.e., we have derived the desired contradiction.
\qed\end{proof}

Now, we are ready to prove Theorem~\ref{thm_mc}.

\begin{proof}
$\pspace$-hardness holds already for $\ltl$ model checking~\cite{SistlaClarke85}, which is a special case of $\pltlc$ model checking. Membership is witnessed by the following algorithm: check whether the colored Büchi graph~$\sys \times \aut$ has an initial pumpable fair path, which is correct due to Lemma~\ref{lemma_mccharacterization}. But as the graph is of exponential size, it has to be constructed and tested for non-emptiness on-the-fly.

Due to Lemma~\ref{lemma_upnonemptiness}, it suffices to check for the existence of an ultimately periodic path~$\pi_0\pi_1^\omega$ such that $\size{\pi_0\pi_1}\le n \in \mathcal{O}(\size{\sys \times \aut}^2)$, i.e., $n$ is  exponential in the size of $\varphi$ and quadratic in the size of $\sys$. To this end, one guesses a vertex $v$ (the first vertex of $\pi_1$) and checks the following reachability properties:
\begin{enumerate}
	\item\label{item_pi1} Is $v$ reachable from $v_I$ via a path where each block contains a cycle with non-zero cost?
	\item Is $v$ reachable from $v$ via a non-empty path that visits an accepting vertex and which either has no changepoint or where each block contains a cycle with non-zero cost? In this case, we also require that $v$ and the last vertex on the path from $v_I$ to $v$ guessed in item \ref{item_pi1}.) differ on their third component in order to make $v$ a changepoint. This spares us from having a block that spans $\pi_0$ and $\pi_1$. 
\end{enumerate}
All these reachability problems can be solved in non-deterministic polynomial space, as a successor of a vertex of $\sys \times \aut$ can be guessed and verified in polymonial time and the length of the paths to be guessed is bounded by $n$, which can be represented with  polynomially many bits. 
\qed\end{proof}

Furthermore, by applying both directions of the proof of Lemma~\ref{lemma_mccharacterization}, we obtain an exponential upper bound on the values of a satisfying variable valuation, if one exists. This is asymptotically tight, as one can already show exponential lower bounds for $\prompt$~\cite{KupfermanPitermanVardi09}.

\begin{corollary}
\label{cor_mcub}
Fix a transition system~$\sys$ and a $\pltlc$-formula~$\varphi$ such that $\sys$ satisfies $\varphi$ with respect to some $\alpha$. Then, $\sys$ satisfies $\varphi$ with respect to a valuation that is bounded exponentially in the size of $\varphi$ and linearly in the number of states of $\sys$ and in the maximal cost in $\sys$.
\end{corollary}

Dually, one can show the existence of an exponential variable valuation that witnesses whether a given $\pltlcg$ specification is satisfied with respect to every variable valuation. The following lemma states the contrapositive of this statement, which we prove using pumping arguments that are similar to the ones for the analogous results for $\pltlg$ and $\pldl_\G$~\cite{FaymonvilleZimmermann15}.

\begin{lemma}
\label{lemma_ubmcalways}
Fix a transition system~$\sys$ and a $\pltlcg$-formula~$\varphi$ such that $\sys$ does not satisfy $\varphi$ with respect to every $\alpha$. Then, $\sys$ does not satisfy $\varphi$ with respect to a valuation that is bounded exponentially in the size of $\varphi$ and linearly in the number of states of $\sys$ and in the maximal cost in $\sys$.
\end{lemma}

\begin{proof}

Let $\aut$ be a Büchi automaton recognizing the models of $\rel{\neg \varphi} \wedge \chi$, which is of exponential size in $\size{\varphi}$. Define $k^* = (4 \cdot \size{\aut} \cdot \size{\sys}+2)\cdot W$, where $W$ is the largest cost in $\sys$, and let $\alpha^*$ be the variable valuation mapping every variable to $k^*$. We consider the contrapositive and show: if there is an $\alpha$ such that $\sys$ does not satisfy $\varphi$ with respect to $\alpha$, then $\sys$ does not satisfy $\varphi$ with respect to $\alpha^*$. 

Thus, assume there is an $\alpha$ and a path~$\pi$ such that $(\trace(\pi), \alpha) \models \neg \varphi$. Due to upwards-monotonicity we can assume w.l.o.g.\ that $\alpha$ maps all variables to the same value, call it $k$. 

Let $\trace(\pi)'$ be a $(k^*+W+1)$-bounded and $(k^*+1)$-spaced $p$-coloring of $\trace(\pi)$ that starts with $p$ not holding true in the first position, which can always be constructed as $W$ is the largest cost. Applying Lemma~\ref{lemma_altcolor}.\ref{lemma_altcolor_pltl2ltl} shows that $\trace(\pi)'$ satisfies $\rel{\neg \varphi}$. Furthermore, it satisfies $\chi$ by construction. Fix some accepting run of $\aut$ on $\trace(\pi)'$ and consider an arbitrary block of $\trace(\pi)'$: if the run does not visit an accepting state during the block, we can remove (if necessary) infixes of the block where the run reaches the same state before and after the infix and where the state of $\sys$ at the beginning and the end of the infix are the same, until the block has length at most $\size{\aut} \cdot \size{\sys}$ and thus cost at most $\size{\aut} \cdot \size{\sys} \cdot W$. 

On the other hand, assume the run visits at least one accepting state during the block. Fix one such position. Then, we can remove infixes as above between the beginning of the block and the position before the accepting state is visited and between the position after the accepting state is reached and before the end of the block. What remains is a block whose length is at most $2\cdot\size{\aut}\cdot\size{\sys}+1$, at it has most $\size{\aut}\cdot \size{\sys}$ many positions before the designated position, this position itself, and at most $\size{\aut}\cdot \size{\sys}$ many after the designated position. Hence, the block has cost at most $(2\cdot\size{\aut}\cdot\size{\sys}+1) \cdot W$.

Thus, we have constructed a $(2 \cdot\size{\aut}\cdot \size{\sys}+1)\cdot W$-bounded $p$-coloring~$\trace(\hat{\pi})'$ of a trace~$\trace(\hat{\pi})$ for some path~$\hat{\pi}$ of $\sys$, as well as an accepting run of $\aut$ on $\trace(\hat{\pi})'$. Hence, $\trace(\hat{\pi})'$ is a model of $\rel{\neg \varphi}$ and applying Lemma~\ref{lemma_altcolor}.\ref{lemma_altcolor_ltl2pltl} shows that $\trace(\hat{\pi})$ is a model of $\neg \varphi$ with respect to the variable valuation mapping every variable to $2\cdot(2\cdot\size{\aut}\cdot \size{\sys}+1) \cdot W=k^*$. Therefore, $\sys$ does not satisfy $\varphi$ with respect to $\alpha^*$.
\qed
\end{proof}

%% file: games.tex
An arena~$\arena = (V, V_0, V_1, v_I, E, \ell, \cst)$ consists of a finite directed graph~$(V, E)$, a partition~$(V_0, V_1)$ of $V$, an initial vertex~$v_I \in V$, a labeling~$\ell \colon V \rightarrow \pow{P}$, and a cost function~$\cst \colon E \rightarrow \nats$. Again, we encode the weights in binary, although the algorithms  we present here and their running times and space requirements  are oblivious to the exact weights. Also, we again assume that every vertex has at least one successor to avoid dealing with finite paths. Also, we again ensure our requirement on the proposition~$\inc$ to indicate the sign of the costs in a cost-trace: if $\inc \in \ell(v')$, then we require $\cst(v,v') >0$ for every edge~$(v,v') \in E$ leading to $v'$. Dually, if $\inc \notin \ell(v')$, then $\cst(v,v') = 0$ for every edge~$(v,v') \in E$.

A play~$\rho = \rho_0 \rho_1 \rho_2 \cdots$ is a path through $\arena$ starting in $v_I$ and its cost-trace~$\trace(\rho)$ is defined as
\[
\trace(\rho)  = 
\ell(\rho_0)\, \cst(\rho_0, \rho_1)\,
\ell(\rho_1)\, \cst(\rho_1, \rho_2)\,
\ell(\rho_2)\, \cst(\rho_2, \rho_3)
\cdots.\]

 A strategy for Player~$i \in \set{0,1}$ is a mapping~$\sigma \colon V^*V_i \rightarrow V$ with $(v, \sigma(wv)) \in E$ for every $w \in V^*$ and $v \in V_i$. A play~$\rho$ is consistent with $\sigma$ if $\rho_{n+1} = \sigma (\rho_0 \cdots \rho_n)$ for every $n$ with $\rho_n \in V_i$.

A $\pltlc$ game~$\game = (\arena, \varphi)$ consists of an arena~$\arena$ and a winning condition~$\varphi$, which is a $\pltlc$ formula. A strategy~$\sigma$ for Player~$0$ is winning with respect to some variable valuation $\alpha$, if the trace of every play that is consistent with $\sigma$ satisfies the winning condition~$\varphi$ with respect to $\alpha$.

We are interested in determining whether Player~$0$ has a winning strategy for a given $\pltlc$~game, and in determining a winning strategy for her if this is the case, which we refer to as solving the game.

\begin{theorem}
\label{thm_games}
Determining whether Player~$0$ has a winning strategy in a given $\pltlc$ game is $\twoexp$-complete. Furthermore, a winning strategy (if one exists) can be computed in doubly-exponential time. 
\end{theorem}

Our proof technique is a generalization of the one for infinite games with $\pltl$ winning conditions~\cite{Zimmermann13}, which in turn extended Kupferman et al.'s solution for the $\prompt$ realizability problem~\cite{KupfermanPitermanVardi09}. First, we note that it is again sufficient to consider $\pltlcf$ formulas, as we are interested in the existence of a variable valuation (see the discussion below Lemma~\ref{lemma_monotonicity}). Next, we apply the alternating-color technique: to this end, we modify the arena to allow Player~$0$ to produce colorings of plays of the original arena and use the relativized winning condition, i.e., we reduce the problem to a game with $\ltl$ winning condition. The winner (and a winning strategy) of such a game can be computed in doubly-exponential time~\cite{PnueliRosner89,PnueliRosner89a}. 

To allow for the coloring, we double the vertices of the arena, additionally label one copy with $p$ and the other not, and split every move into two: first, the player whose turn it is picks an outgoing edge, then Player~$0$ decides in which copy she wants to visit the target, thereby picking the truth value of $p$.

Formally, given an arena $\arena=(V,V_0,V_1, v_I, E,\ell, \cst)$, the extended arena $\arena' =
(V',V_0',V_1', v_I', E',\ell', \cst')$ consists of 
\begin{itemize}
\item $V'=V\times\{0,1\}\cup E$,
\item $V_0'=V_0\times\{0,1\}\cup E$ and $V_1'=V_1\times\{0,1\}$,
\item $v_I' = (v_I, 0)$,
\item $E'=\{((v,0),e),((v,1),e),(e,(v',0)),(e,(v',1))\mid e=(v,v')\in E\}$,
\item $\ell'(e)=\emptyset$ for every $e\in E$ and
$\ell'(v,b)=\begin{cases}\ell(v) & \text{if $b=0$},\\
{\ell(v)\cup\{p\}}&\text{if $b=1$},\\
         \end{cases}$ and
\item $\cst'((v,b),(v,v')) = \cst(v,v') $ and $\cst'((v,v'),(v',b'))=0$.
\end{itemize}

A path through the new arena~$\arena'$ has the
form $(\rho_0,b_0)e_0(\rho_1,b_1)e_1(\rho_2,b_2)\cdots$ for some path~$\rho_0\rho_1\rho_2\cdots$ through $\arena$, where  $e_n=(\rho_n,\rho_{n+1})$ and $b_n \in \{0,1\}$. Also, we have
$\card{\arena'}\in\mathcal{O}(\card{\arena}^2)$. Note that we use the costs in $\arena'$ only to argue the correctness of our construction, not to define the winning condition for the game in $\arena'$. 
 
Also, note that the additional choice vertices of the form $e \in E$ have to be ignored when it comes to evaluating the winning condition on the trace of a play. Thus, we consider games with $\ltl$ winning conditions under so-called \emph{blinking semantics}: Player~$0$ wins a play~$\rho = \rho_0 \rho_1 \rho_2 \cdots$ under blinking semantics, if $\ell(\rho_0)\ell(\rho_2)\ell(\rho_4) \cdots$ satisfies the winning condition~$\varphi$; otherwise, Player~$1$ wins. Winning strategies under blinking semantics are defined as expected. Determining whether Player~$0$ has a winning strategy for a given game with $\ltl$ winning condition under blinking semantics is $\twoexp$-complete, which can be shown by a slight variation of the proof for $\ltl$ games under classical semantics~\cite{PnueliRosner89,PnueliRosner89a}. Furthermore, if Player~$0$ has a winning strategy for such a game, then also a finite-state one of at most doubly-exponential size in $\size{\varphi}$. 

Such a strategy for an arena $(V,V_0,V_1, v_I, E,\ell, \cst)$ is given by a memory structure~$\mem = (M, m_I, \update)$ with a finite set~$M$ of memory states, an initial memory state~$m_I \in M$, and an update function~$\update\colon M \times V \rightarrow M$, and by a next-move function~$\nextmove \colon V_0 \times M \rightarrow V$ satisfying $(v, \nextmove(v,m)) \in E$ for every $m$ and every $v$. The function~$\update^* \colon V^+ \rightarrow M$ is defined via $\update^*(v) = m_I$ and $\update^*(wv) = \update(\update^*(w),v)$. Then, the strategy~$\sigma$ implemented by $\mem$ and $\nextmove$ is defined by $\sigma(wv) = \nextmove(v, \update^*(wv))$. The size of $\sigma$ is (slightly abusively) defined as $\size{M}$.

Given a game~$(\arena, \varphi)$ with $\pltlcf$ winning condition~$\varphi$, define $\arena'$ as above and let $\varphi' = \rel{\varphi} \wedge \chi$, where 
$\chi = (\G\F p \wedge \G\F \neg p) \leftrightarrow \G\F \inc$. 
Recall that $\chi$ is satisfied by a cost-trace, if the trace has infinitely many changepoints if and only if it has cost~$\infty$.
 
\begin{lemma}
\label{lemma_gamescharacterization}[cp.\ Lemma 3.1 of \cite{KupfermanPitermanVardi09}]
Player~$0$ has a winning strategy for $(\arena, \varphi)$ with respect to some $\alpha$ if and only if she has a winning strategy for $(
\arena',\varphi')$ under blinking semantics.
\end{lemma}

\begin{proof}
Let $\sigma$ be a winning strategy for Player~$0$ in $(\arena, \varphi)$ with respect to some fixed $\alpha$ and define $k = \max_{x \in \var(\varphi)}\alpha(x)$.
We define a strategy~$\sigma'$ for $\arena'$ as follows:
\begin{equation*}
\sigma'((\rho_0,b_0)(\rho_0, \rho_1) \cdots (\rho_{n-1}, \rho_n)
(\rho_n,b_n))=(\rho_n,\sigma(\rho_0\cdots\rho_n))
\end{equation*}
if $(\rho_n, b_n) \in V_0'$, which implies $\rho_n \in V_0$. Thus, at a 
non-choice vertex, Player~$0$ mimics the behavior of $\sigma$. At choice 
vertices, she alternates between the two copies of the arena every time  the cost has exceeded $k+1$: let \[w = (\rho_0,b_0)(\rho_0, \rho_1) \cdots
(\rho_n,b_n) (\rho_n, \rho_{n+1})\] be a play prefix ending in a choice vertex and let $n' \le n$ be the last changepoint in $\ell'(\rho_0, b_0) \cdots \ell'(\rho_n, b_n)$. Now, we define
\begin{equation*}
\sigma'(w)=\begin{cases}
(\rho_{n+1},0) & \text{if ($\cst(\rho_{n'} \cdots \rho_n) < k+1$ and $b_n = 0$) or}\\
&\text{($\cst(\rho_{n'} \cdots \rho_n) \ge k+1$ and $b_n = 1$),}\\
(\rho_{n+1},1) & \text{if ($\cst(\rho_{n'} \cdots \rho_n) < k+1$ and $b_n = 1$) or}\\
&\text{($\cst(\rho_{n'} \cdots \rho_n) \ge k+1$ and $b_n = 0$).}
\end{cases}
\end{equation*}

Let $\rho=\rho_0\rho_1\rho_2\cdots$ be a play in $\arena'$ that is consistent
with $\sigma'$ and let
\begin{equation*}
\rho'=\rho_0\rho_2\rho_4\cdots=(v_0,b_0)(v_1,b_1)(v_2,b_2)\cdots.
\end{equation*}
By definition of $\sigma'$, the sequence~$v_0v_1v_2\cdots$ is a play in $\arena$
that is consistent with $\sigma$ and thus winning for Player~$0$ with respect to
$\alpha$, i.e.,  $(\trace(v_0v_1v_2\cdots), \alpha)\models\varphi$.
Also, $w' = \ell'(v_0, b_0)\ell'(v_1, b_1)\ell'(v_2, b_2) \cdots$ is a $(k+1)$-spaced coloring of the trace
$\trace(v_0v_1v_2\cdots)$. Hence, $w'\models \varphi'$ due to
Lemma~\ref{lemma_altcolor}.\ref{lemma_altcolor_pltl2ltl}. Finally, $w'$ satisfies $\chi$ by construction. Thus,
$\sigma'$ is a winning strategy for $(\arena', \varphi')$
under blinking semantics.

Now, let $\sigma'$ be a winning strategy for Player~$0$ in $(\arena', \varphi')$ which we can assume (w.l.o.g.) to be implemented by $\mem'=(M',m_I',\update')$ and some next-move function~$\nextmove'$
 such that $\size{M'}$ is doubly-exponential in $\size{\varphi}$. We define a strategy~$\sigma$ for $\arena$ by simulating a play in $\arena'$ that is consistent with $\sigma'$.

To this end, define the memory structure~$\mem = (M,m_I,\update)$ for $\arena$ with $M = (V\times\{0,1\})\times
M'$, $m_I= ((v,0),m_I')$, and
\begin{equation*}
\update(((v,b),m),v') = (\nextmove'(e, m'),\update'(m',\nextmove'(e,m')))
\end{equation*}
where $e = (v, v')$ and $m'= \update'(m,e)$. Intuitively, the update-function mimics two moves in $\arena'$: first, the one from $(v,b)$ to $e = (v,v')$ and then the move from this choice vertex determined by the strategy $\sigma'$, which is given by $\nextmove'(e,m')$, where $m'$ is the updated memory state.

Let $w$ be a play prefix of a play in $\arena$. The memory state~$\update^* (w)
= ((v,b),m)$ encodes the following information: the simulated play~$w'$ in
$\arena'$ ends in $(v, b)$, where $v$ is the last vertex of $w$, and we have
$\update'^*(w') = m$. Hence, it contains all information necessary to apply the
next-move function $\nextmove'$ to mimic $\sigma'$. Thus, we define the
next-move function~$\nextmove \colon V_0\times M\rightarrow V$ for Player~$0$ in
$\arena$ by
\begin{equation*}
\nextmove(v,((v',b),m))=
  \begin{cases}
                       v''&\text{if $v= v'$ and }
\nextmove'((v',b),m)=(v',v''),\\
                       \overline{v} &\text{otherwise, for some $\overline{v} \in
V$ with $(v,\overline{v})\in E$. }
                       \end{cases}
\end{equation*}
By definition of $\mem$, the second case of the definition is never invoked,
since $\update^*(wv) = ((v',b),m)$ always satisfies $v = v'$.

It remains to show that the strategy~$\sigma$ implemented by $\mem$ and
$\nextmove$ is indeed a winning strategy for Player~$0$ for $(\arena,\varphi)$
with
respect to some $\alpha$.  To this end, let $k = (\card{V}\cdot\card{M}+3)\cdot W$ and define $\alpha(x)=2k$ for every $x$, where $W$ is the largest weight in $\arena$. 

Let $\rho_0\rho_1\rho_2\cdots$ be a play in $\arena$ that is consistent with
$\sigma$. A straightforward induction shows that there exist
bits~$b_0,b_1,b_2,\cdots$ such that the play
\[(\rho_0,b_0)(\rho_0,\rho_1)(\rho_1,b_1)(\rho_1,\rho_2)(\rho_2, b_2)\cdots\] in
$\arena'$ is consistent with $\sigma'$.
 Hence, $w'' = \ell'(\rho_0,b_0)\ell'(\rho_1,b_1)\ell'(\rho_2,b_2)\cdots$
satisfies $\varphi'$. We show that
$w''$ is $k$-bounded. This
suffices to finish the proof as we can  apply
Lemma~\ref{lemma_altcolor}.\ref{lemma_altcolor_ltl2pltl} and
obtain $(\trace(\rho),\alpha)\models\varphi$, as $w''$ is a $k$-bounded
coloring of $\trace(\rho)$. Thus, $\sigma$ is  a winning strategy for
Player~$0$ for $(\arena, \varphi)$ with respect to $\alpha$.

Towards a contradiction assume that $w''$ is not $k$-bounded. Then,
there exist positions~$i<j$ such that
\begin{itemize}
	\item $\rho_{i}=\rho_{j}$,
	\item $\update'^*((\rho_0,b_0)\cdots
(\rho_{i},b_{i}))=\update'^*((\rho_0,b_0)\cdots
(\rho_{j},b_{j}))$,
	\item the bits~$b_i, \ldots, b_j$ are all equal, and
	\item $\cst(\rho_i \cdots \rho_j) > 0$.
\end{itemize}
To show this, one defines the sets~$V_j$ of vertices visited between the  $j$-th and the $(j+1)$-th edge with non-zero cost (including the $j$-th edge). This yields $\size{V}\cdot\size{M}+1$ non-empty sets of vertices of $(V\times\set{0,1}) \times M$ that coincide on the bit stored in their second component. Hence, we have derived the desired vertex repetition, as there are only $\size{V}\cdot\size{M}$ such vertices. 

Thus, the play \[\rho^* =
(\rho_0,b_0)\cdots
(\rho_{i-1},b_{i-1})\big[(\rho_{i},b_{i})\cdots(\rho_{j-1},b_{j-1}
)(\rho_{j-1},\rho_{j})\big]^{\omega},\] obtained by traversing the
cycle between $(\rho_{i},b_{i})$ and $(\rho_{j},b_{j})$ infinitely often, is
consistent with $\sigma'$, since the memory states reached at the beginning and
the end of the loop are the same. Remember that the bits do not change between
$i$ and $j$. Thus, $\trace(\rho^*)$ has only finitely many change points, but infinitely many occurrences of $\inc$ and
does therefore not satisfy $\chi$ under blinking semantics. This contradicts the fact
that $\sigma'$ is a winning strategy for $(\arena', \rel{\varphi} \wedge
\chi)$ under blinking semantics.
\qed\end{proof}

Now, we are able to prove Theorem~\ref{thm_games}.

\begin{proof}
Hardness follows immediately from the $\twoexp$-hardness of determining the winner of an $\ltl$ game~\cite{PnueliRosner89,PnueliRosner89a}, which is a special case.

Membership in $\twoexp$ follows from the reductions described above: first, we turn the winning condition into a $\pltlcf$ formula and construct the $\ltl$ game under blinking semantics obtained from expanding the arena and relativizing the winning condition. This game is only polynomially larger than the original one and its winner (and a winning strategy) is computable in doubly-exponential~time. 
\qed\end{proof}

By applying both directions of the proof of Lemma~\ref{lemma_gamescharacterization}, we obtain a doubly-exponential upper bound on the values of a satisfying variable valuation, if one exists. This is asymptotically tight, as one can already show doubly-exponential lower bounds for $\prompt$~\cite{Zimmermann13}.

\begin{corollary}
\label{cor_gamesub}
Fix a $\pltlc$ game~$\game = (\arena, \varphi)$ such that Player~$0$ has a winning strategy for $\game$ with respect to some $\alpha$. Then, Player~$0$ has a winning strategy for $\game$ with respect to a valuation that is bounded doubly-exponentially in the size of $\varphi$ and linearly in the number of vertices of $\arena$ and in the maximal cost in $\arena$.
\end{corollary}

%% file: cpldl.tex
Linear Dynamic logic ($\ldl$)~\cite{GiacomoVardi13,Vardi11} extends $\ltl$ by temporal operators guarded with regular expressions, e.g., $\ddiamond{r}\varphi$ holds at position~$n$, if there is a $j$ such that $\varphi$ holds at position~$n+j$ and the infix between positions~$n$ and $n+j-1$ matches $r$. The resulting logic has the full expressiveness of the $\omega$-regular languages while retaining many of $\ltl$'s desirable properties like a simple syntax, intuitive semantics, a polynomial space algorithm for model checking, and a doubly-exponential time algorithm for solving games. Parametric $\ldl$ ($\pldl$)~\cite{FaymonvilleZimmermann14} allows to parameterize such operators, i.e., $\ddiamondle{r}{x}\varphi$ holds at position~$n$ with respect to a variable valuation~$\alpha$, if there is a $j \le \alpha(x)$ such that $\varphi$ holds at position~$n+j$ and the infix between positions~$n$ and $n+j-1$ matches $r$. Model checking and solving games with $\pldl$ specifications is not harder than for $\ltl$, although $\pldl$ is more expressive and has parameterized operators. In this section, we consider $\pldlc$ where the parameters bound the cost of the infix instead of the length. 

Formally, formulas of $\pldlc$ are given by the grammar
\begin{align*}
\varphi &\cceq p \mid \neg p \mid \varphi \wedge \varphi \mid \varphi \vee \varphi
  \mid \ddiamond{r} \varphi 
  \mid \bbox{r} \varphi 
  \mid \ddiamondle{r}{z} \varphi 
  \mid \bboxle{r}{z} \varphi\\
  r & \cceq \phi \mid \varphi? \mid r+r \mid r \conc r \mid r^*
\end{align*}
where $p \in P$, $z \in \Var$, and where $\phi$ ranges over propositional formulas over $P$. As for $\pltlc$, $\pldlc$ formulas are evaluated on cost-traces with respect to variable valuations. Satisfaction of atomic formulas and of conjunctions and disjunctions is defined as usual, and for the four temporal operators, we define
\begin{itemize}

\item $(w, n, \alpha) \models \ddiamond{r}\varphi$ if there exists $j \ge 0$ such that $(n, n+j) \in \Rexp(r, w, \alpha)$ and $(w, n+j, \alpha) \models \varphi$, 

\item $(w, n, \alpha) \models \bbox{r}\varphi$ if for all $j \ge 0$ with $(n, n+j) \in \Rexp(r, w, \alpha)$ we have $(w, n+j, \alpha) \models \varphi$,

\item $(w, n, \alpha) \models \ddiamondle{r}{z}\varphi$ if there exists $j \ge 0$ with $\cst(w_n c_n \cdots c_{n+j-1}w_{n+j}) \le \alpha(z)$ such that $(n, n+j) \in \Rexp(r, w, \alpha)$ and $(w, n+j, \alpha) \models \varphi$, and

\item $(w, n, \alpha) \models \bboxle{r}{z}\varphi$ if for all $j \ge 0$ with $\cst(w_n c_n \cdots c_{n+j-1}w_{n+j}) \le \alpha(z)$ and with $(n, n+j) \in \Rexp(r, w, \alpha)$ we have $(w, n+j, \alpha) \models \varphi$.

\end{itemize}
Here, the relation~$\Rexp(r,w,\alpha) \subseteq \nats\times\nats$ contains all pairs~$(m,n) \in \nats \times \nats$ such that $w_m \cdots w_{n-1}$ matches $r$ and is defined inductively by 
\begin{itemize}
\item $\Rexp(\phi,w,\alpha) = \set{(n, n+1) \mid w_n \models \phi}$ for propositional~$\varphi$,
\item $\Rexp(\psi?,w,\alpha) = \set{(n, n) \mid (w, n, \alpha) \models \psi}$,
\item $\Rexp(r_0 + r_1, w, \alpha) = \Rexp(r_0, w, \alpha) \cup \Rexp(r_1, w, \alpha)$,
\item $\Rexp(r_0 \conc r_1, w, \alpha) = \set{(n_0, n_2) \mid \exists n_1 \text{ s.t. }(n_0,n_1)\in \Rexp(r_0, w, \alpha) \text{ and }$ \newline $ (n_1, n_2) \in \Rexp(r_1, w, \alpha)}$, and
\item $\Rexp(r^*, w, \alpha) = \set{(n,n) \mid n\in\nats} \cup \set{(n_0, n_{k+1}) \mid \exists n_1, \ldots, n_{k} \text{ s.t. } $ \newline $ (n_j, n_{j+1}) \in \Rexp(r, w, \alpha) \text{ for all } j \le k}$.
\end{itemize}
Again, we restrict ourselves to well-formed formulas, i.e., those whose set of variables parameterizing diamond operators and whose set of variables parameterizing box operators are disjoint. 

Using the duality of the operators~$\ddiamondle{r}{z}$ and $\bboxle{r}{z}$ (note that $r$ is left unchanged), one can prove an analogue of Lemma~\ref{lemma_pltlnegation}.

\begin{lemma}
\label{lemma_pldlnegation}
For every $\pldlc$ formula~$\varphi$ there exists an efficiently constructible
$\pldlc$ formula~$\neg \varphi$ s.t.\
\begin{enumerate}
\item $(w,n,\alpha)\models \varphi$ if and only if $(w,n,
\alpha) \not\models \neg \varphi$ for every $w$, every $n$, and every $\alpha$, and
\item $\card{\neg \varphi} =  \card{\varphi}$.
\end{enumerate}
\end{lemma}

Note that we do not claim that negation preserves well-formedness  and that we have not (yet) defined unipolar fragments of $\pldlc$. This is because the former statement is wrong: the negation of the well-formed $\pldlc$-formula $\bboxle{(\bboxle{p}{x}p)?}{x}p$ is $\ddiamondle{(\bboxle{p}{x}p)?}{x}\neg p$, which is not well-formed. The issue is that negation does not flip the duality of parameterized operators in tests, i.e., formulas of the form~$\phi?$ in regular expressions, which also requires us to be careful when defining the unipolar fragments of $\pldlc$: let $\varphi$ be a $\pldlc$ formula.
\begin{itemize}
	\item $\varphi$ is a $\pldlcdiamond$ formula, if it does not contain a parameterized box operator. 
	\item $\varphi$ is a $\pldlcbox$ formula, if it does not contain a parameterized diamond operator \emph{and} if its negation is a $\pldlcdiamond$ formula. 
\end{itemize}
For $\pltlcg$ formulas, the second conjunct in the second item above is trivial, but, as we have seen in the example above, this is no longer true for $\pldlcbox$. The second conjunct is necessary to be able to solve problems for $\pldlcbox$ by dualizing the formula into an $\pldlcdiamond$ formula. This becomes crucial when we consider the optimization problems in Setion~\ref{sec_opt}, the only place where we deal with $\pldlcbox$ formulas.

Finally, Lemma~\ref{lemma_monotonicity} holds for $\pldlc$, too. 

\begin{lemma}
\label{lemma_monotonicity_pldlc}
Let $\varphi$ be a $\pldlc$ formula and let $\alpha$ and
$\beta$ be variable valuations satisfying $\alpha ( x ) \le \beta ( x)$ for
every $x \in \varF( \varphi)$ and $\alpha ( y ) \ge \beta ( y)$ for every $y \in \varG( \varphi)$. If $(w, \alpha) \models \varphi$, then $(w, \beta) \models
\varphi$.
\end{lemma}

The alternating-color technique is applicable to $\pldl$~\cite{FaymonvilleZimmermann14}: to this end, one introduces changepoint-bounded variants of the unparameterized diamond operator and of the unparameterized box operator whose semantics only quantify over infixes with at most one changepoint. $\ldl$ formulas with  changepoint-bounded operators can be translated into Büchi automata of exponential size. As usual, the parameterized box operators can again be disregard due to monotonicity. Thus, given a $\pldldiamond$ formula, one replaces every diamond operator by a changepoint-bounded one and can then show that both formulas are equivalent, provided the distance between color-changes is appropriately bounded and spaced. This allows to extend the algorithms for model checking and realizability based on the alternating-color technique to $\pldl$. The detailed construction is described in~\cite{FaymonvilleZimmermann14}.

In the setting with costs investigated here, the approach is similar: one has to replace the parameterized diamond operators by changepoint-bounded ones. Furthermore, the formula~$\chi = (\G\F p \wedge \G\F \neg p) \leftrightarrow \G\F \inc$ used in the applications of the alternating-color technique in Sections~\ref{sec_mc} and \ref{sec_games} is replaced by an equivalent $\ldl$ formula, which is possible as $\ldl$ subsumes $\ltl$. The resulting formula is again translatable into a Büchi automaton of exponential size. Thus, the constructions presented in the previous two sections solving the model checking and the game problem are again applicable.

\begin{theorem}
The $\pldlc$ model checking problem is $\pspace$-complete
and solving infinite games with $\pldlc$ winning conditions is $\twoexp$-complete.
\end{theorem}

%% file: multcost.tex
In this section, we consider parameterized temporal logics with  multiple cost-functions. For the sake of simplicity, we restrict our attention to $\pltlc$, although all results hold for $\pldlc$, too.

Fix some dimension~$d \in \nats$. The syntax of $\mpltlc$ is obtained by equipping the parameterized temporal operators by a coordinate~$i \in \set{1, \ldots, d}$, denoted by $\F_{\le_i x}$ and $\G_{\le_i y}$. Here, a cost-trace is of the form
$
w_0\, \overline{c}_0\, 
w_1\, \overline{c}_1\, 
w_2\, \overline{c}_2\, 
\cdots$
with $w_n \in \pow{P}$ and $\overline{c}_n \in \nats^d$. Thus, for every $i \in \set{1, \ldots, d}$, we can define \[\cst_i(w_0 \overline{c}_0 \cdots \overline{c}_{n-1} w_{n}) = \sum_{j=0}^{n-1} (\overline{c}_j)_i\] for every finite cost-trace~$w_0 \overline{c}_0 \cdots \overline{c}_{n-1} w_{n}$, where $(\overline{c}_j)_i$ denotes the $i$-th entry of the vector~$\overline{c}_j$. Furthermore, we require for every coordinate~$i$ a proposition~$\inc_i$ such that $\inc_i \in w_{n+1}$ if and only if $(\overline{c}_n)_i >0$. 

The semantics of atomic formulas, boolean connectives, and unparameterized temporal operators are unchanged and for the parameterized operators, we define 
\begin{itemize}
	\item $(w,n,\alpha)\models\F_{\le_i z}\varphi$ if and only if there exists a $j \ge 0$ with\newline $\cst_i(w_{n} \overline{c}_n \cdots \overline{c}_{n+j-1} w_{n+j}) \le \alpha(z)$ such that $(w,n+j,\alpha)\models\varphi$, and

\item $(w,n,\alpha)\models\G_{\le_i z}\varphi$ if and only if for every $j \ge 0$ with \newline$\cst_i(w_{n} \overline{c}_n \cdots \overline{c}_{n+j-1} w_{n+j}) \le \alpha(z)$: $(w,n+j,\alpha)\models\varphi$.

\end{itemize}

In this setting, we consider the model checking problem for transition systems with $d$ cost functions and want to solve games in arenas with $d$ cost functions.

\begin{example}
A Streett condition with costs~$(Q_i, P_i)_{i \in \{1, \ldots, d\}}$~\cite{FijalkowZimmermann14} can be expressed\footnote{The same disclaimer as for the parity condition with costs applies here, see Footnote~\ref{footnote_costparity}.} in $\mpltlc$ via
\[
\F\G  \left( \bigwedge\nolimits_{i \in \set{1, \ldots, d}} \left(Q_i \rightarrow  \F{}_{\le_i x}\, P_i 
\right)\right).
\]
\end{example}

Again, we restrict ourselves to formulas where no variable parameterizes an eventually- and an always operator, but we allow a variable to parameterize operators with different coordinates. Also, the fragments~$\mpltlcf$ and $\mpltlcg$ are defined as for $\pltlc$, i.e., a formula is a $\mpltlcf$ formula (a $\mpltlcg$ formula), if it does not contain parameterized always operators (parameterized eventually operators).

Lemma~\ref{lemma_pltlnegation} can be extended to $\mpltlc$ by adding the rules~$\neg (\F_{\le_i x}\varphi ) = \G_{\le_i x} \neg \varphi$ and $\neg (\G_{\le_i y}\varphi ) = \F_{\le_i y} \neg \varphi$ to the proof.

\begin{lemma}
\label{lemma_mpltlcnegation}
For every $\mpltlc$ formula~$\varphi$ there exists an efficiently constructible
$\mpltlc$ formula~$\neg \varphi$ s.t.\
\begin{enumerate}
\item $(w,n,\alpha)\models \varphi$ if and only if $(w,n,
\alpha) \not\models \neg \varphi$ for every $w$, every $n$, and every $\alpha$,
\item $\card{\neg \varphi} =  \card{\varphi}$.
\item If $\varphi$ is well-formed, then so is $\neg \varphi$.
\item If $\varphi$ is an $\ltl$ formula, then so is $\neg \varphi$.
\item If $\varphi$ is a $\mpltlcf$ formula, then $\neg \varphi$ is a $\mpltlcg$ formula
and vice versa.
\end{enumerate}
\end{lemma}

Furthermore, Lemma~\ref{lemma_monotonicity} holds for $\mpltlc$ as well. 

\begin{lemma}
\label{lemma_monotonicity}
Let $\varphi$ be a $\pltlc$ formula and let $\alpha$ and
$\beta$ be variable valuations satisfying $\alpha ( x ) \le \beta ( x)$ for
every $x \in \varF( \varphi)$ and $\alpha ( y ) \ge \beta ( y)$ for every $y \in \varG( \varphi)$. If $(w, \alpha) \models \varphi$, then $(w, \beta) \models
\varphi$.
\end{lemma}

The alternating-color technique is straightforwardly extendable to the new logic $\mpltlc$: one introduces a fresh proposition~$p_i$ for each coordinate~$i$ and defines~$\chi = \bigwedge_{i=1}^d ( (\G\F p_i \wedge \G\F \neg p_i) \leftrightarrow \G\F \inc_i )$. Furthermore, the notions of $i$-blocks, $k$-boundedness in coordinate~$i$, and $k$-spacedness in coordinate~$i$ are defined as expected. Then, the proofs presented in Section~\ref{sec_mc} and Section~\ref{sec_games} can be extended to  the setting with multiple cost-functions.

In the model checking case, the third component of the set of states of the colored Büchi graph~$\sys \times \aut$ has the form~$\pow{\set{p_1, \ldots, p_d}}$, i.e., it is of exponential size. However, this is no problem, as the automaton~$\aut$ is already of exponential size. Similarly, in the case of infinite games, each vertex of the original arena has $2^d$ copies in $\arena'$, one for each element in $\pow{\set{p_1, \ldots, p_d}}$ allowing Player~$0$ to produce appropriate colorings with the propositions~$p_i$. The resulting game has an arena of exponential size (in the size of the original arena and of the original winning condition) and an $\ltl$ winning condition under blinking semantics. Such a game can still be solved in doubly-exponential time. To this end, one turns the winning condition into a deterministic parity automaton of doubly-exponential size with exponentially many colors, constructs the product of the arena and the parity automaton, which yields a parity game of doubly-exponential size with exponentially many colors. Such a game can be solved in doubly-exponential time~\cite{Schewe07}. 

\begin{theorem}
The $\mpltlc$ model checking problem is $\pspace$-complete and solving infinite games with $\mpltlc$ winning conditions is $\twoexp$-com\-plete.
\end{theorem}

Again, the same results hold for $\mpldlc$, which is defined as expected. 

%% file: optimization.tex
It is natural to treat model checking and solving games with specifications in parameterized linear temporal logics as an optimization problem: determine the \emph{optimal} variable valuation such that the system satisfies the specification with respect to it. For parameterized eventualities, we are interested in minimizing the waiting times while for parameterized always', we are interested in maximizing the waiting times. Due to the undecidability results for not well-defined formulas one considers the optimization problems for the unipolar fragments, i.e., for formulas having either no parameterized eventualities or no parameterized always'. In this section, we present algorithms for such optimization problems given by $\pltlc$ specifications. In the following, we encode the weights of the transition system or arena under consideration in unary to obtain our results. Whether these results can also be shown for a binary encoding is an open question. 

For model checking, we are interested in the following four problems: given a transition system~$\sys$ and a $\pltlcf$ formula~$\varphi_{\F}$ and a $\pltlcg$ formula~$\varphi_{\G}$, respectively, determine
\begin{enumerate}
\item\label{item_minminmc} $\min_{\{\alpha\mid \text{$\sys$ satisfies $\varphi_{\F}$ w.r.t.\ $\alpha$}\}}\min_{x\in \varF(\varphi_{\F})}\alpha(x)$,

\item $\min_{\{\alpha\mid \text{$\sys$ satisfies $\varphi_{\F}$ w.r.t.\ $\alpha$}\}}\max_{x\in \varF(\varphi_{\F})}\alpha(x)$,

\item $\max_{\{\alpha\mid \text{$\sys$ satisfies $\varphi_{\G}$ w.r.t.\ $\alpha$}\}}\max_{y\in \varG(\varphi_{\G})}\alpha(y)$, and

\item $\max_{\{\alpha\mid \text{$\sys$ satisfies $\varphi_{\G}$ w.r.t.\ $\alpha$}\}}\min_{y\in \varG(\varphi_{\G})}\alpha(y)$.

\end{enumerate}

Applying the monotonicity of the parameterized operators and (in the first case) the alternating-color technique to all but one variable reduces the four optimization problems to ones where the specification has a single variable (cp.~\cite{AlurEtessamiLaTorrePeled01}). Furthermore, the upper bounds presented in Corollary~\ref{cor_mcub} and in Lemma~\ref{lemma_ubmcalways} yield an exponential search space for an optimal valuation: if this space is empty, then there is no $\alpha$ such that $\sys$ satisfies $\varphi_{\F}$ with respect to $\alpha$ in the first two cases. On the other hand, if the search space contains every such $\alpha$, then $\sys$ satisfies $\varphi_{\G}$ with respect to every $\alpha$ in the latter two cases. 

Thus, it remains the check whether the specification is satisfied with respect to some valuation that is bounded exponentially. In this setting, one can construct an exponentially sized non-deterministic Büchi automaton recognizing the models of the specification with respect to the given valuation (using a slight adaption of the construction presented in~\cite{Zimmermann13} accounting for the fact that we keep track of cost instead of time). This automaton can be checked for non-emptiness in polynomial space using an on-the-fly construction. Thus, an optimal $\alpha$ can be found in polynomial space by binary search.

\begin{theorem}
The $\pltlc$ model checking optimization problems can be solved in polynomial space, if the weights are encoded in unary.
\end{theorem}

A similar approach works for infinite games as well. Here, we are interested in computing 
\begin{enumerate}
\item\label{item_minmingame} $\min_{\{\alpha\mid \text{Pl.
 $0$ has winning strategy for $\game_{\F}$ w.r.t.\ $\alpha$}\}}\min_{x\in \varF(\varphi_{\F})}\alpha(x)$,

\item $\min_{\{\alpha\mid \text{Pl.
 $0$ has winning strategy for $\game_{\F}$ w.r.t.\ $\alpha$}\}}\max_{x\in \varF(\varphi_{\F})}\alpha(x)$,

\item $\min_{\{\alpha\mid \text{Pl.
 $0$ has winning strategy for $\game_{\G}$ w.r.t.\ $\alpha$}\}}\min_{x\in \varG(\varphi_{\G})}\alpha(x)$, and

\item $\min_{\{\alpha\mid \text{Pl.
 $0$ has winning strategy for $\game_{\G}$ w.r.t.\ $\alpha$}\}}\max_{x\in \varG(\varphi_{\G})}\alpha(x)$.

\end{enumerate}
and witnessing winning strategies for given $\pltlc$ games $\game_{\F}$ with $\pltlcf$ winning condition $\varphi_{\F}$ and $\game_{\G}$ with $\pltlcg$ winning condition $\varphi_{\G}$.

Again, one can reduce these problems to  the case of winning conditions with a single variable and by applying determinacy of games with respect to a fixed valuation, it even suffices to consider the case of $\pltlcf$ winning conditions with a single variable, due to duality of games: swapping the players in a game with $\pltlcg$ winning condition yields a game with $\pltlcf$ winning condition. Corollary~\ref{cor_gamesub} gives a doubly-exponential upper bound on an optimal variable valuation. Hence, one can construct a deterministic parity automaton of  triply-exponential size with exponentially many colors recognizing the models of the specification with respect to a fixed variable valuation~$\alpha$ that is below the upper bound (again, see~\cite{Zimmermann13} for the detailed construction). Player~$0$ wins the parity game played in the original arena but using the language of the automaton as winning condition if and only if she has a winning strategy for the $\pltlcf$ game with respect to $\alpha$. Such a parity game can be solved in triply-exponential time~\cite{Schewe07}.

\begin{theorem}
The $\pltlc$ optimization problems for infinite games can be solved in triply-exponential time, if the weights are encoded in unary.
\end{theorem}

Furthermore, the same results hold for $\pldlc$ using appropriate adaptions of the automata constructions presented in~\cite{FaymonvilleZimmermann14,FaymonvilleZimmermann15}. Here, we apply the requirement on negations of $\pldlcbox$ formulas being $\pldlcdiamond$ formulas is applied when dualizing a game with $\pldlcbox$ winning condition into a game with $\pldlcdiamond$ winning condition by swapping the players and negating the winning condition. Without this requirement, we would not necessarily end up with a $\pldlcdiamond$ winning condition, but possibly with a non-wellformed winning condition. 

\begin{theorem}
The $\pldlc$	  model checking optimization problems can be solved in polynomial space and the $\pldlc$ optimization problems for infinite games can be solved in triply-exponential time, if the weights are encoded in unary.
\end{theorem}

However, for parameterized logics with multiple cost-functions, these results do not remain valid, as one cannot reduce the optimization problems to ones with a single variable, as a variable may bound operators in different dimensions. Thus, one has to keep track multiple costs, which incurs an additional exponential blow-up when done naively. Whether this can be improved is an open question. 

%% file: conc.tex
\begin{wrapfigure}{r}{.50\textwidth}
\vspace{-25pt}
\begin{center}
\begin{tikzpicture}

\node[rounded corners, draw, black, ultra thick, fill=gray!20] at (0, 0) (ltl) {LTL};
\node[rounded corners, draw, black, ultra thick, fill=gray!20] at (-1, 1) (ldl) {LDL};
\node[rounded corners, draw, black, ultra thick, fill=gray!20] at (1, 1) (pltl) {PLTL};
\node[rounded corners, draw, black, ultra thick, fill=gray!20] at (0, 2) (pldl) {{PLDL}};

\node[rounded corners, draw, black, ultra thick, fill=gray!20] at (1, 3) (cpldl) {{cPLDL}};
\node[rounded corners, draw, black, ultra thick, fill=gray!20] at (2, 2) (cpltl) {{cPLTL}};
\node[rounded corners, draw, black, ultra thick, fill=gray!20] at (2, 4) (mcpldl) {{mult-cPLDL}};
\node[rounded corners, draw, black, ultra thick, fill=gray!20] at (3, 3) (mcpltl) {{mult-cPLTL}};

\path[ultra thick, black,-stealth]
(ltl) edge (ldl)
(ltl) edge  (pltl)
(ldl) edge (pldl)
(pltl) edge (pldl);

\path[ultra thick, black,-stealth]
(pltl) edge (cpltl)
(pldl) edge (cpldl)
(cpltl) edge (cpldl)
(cpltl) edge (mcpltl)
(cpldl) edge (mcpldl)
(mcpltl) edge (mcpldl)
;

\end{tikzpicture}
\end{center}
\caption{Overview over the logics considered in this work.}
\label{fig_logics}
\vspace{-13pt}
\end{wrapfigure}
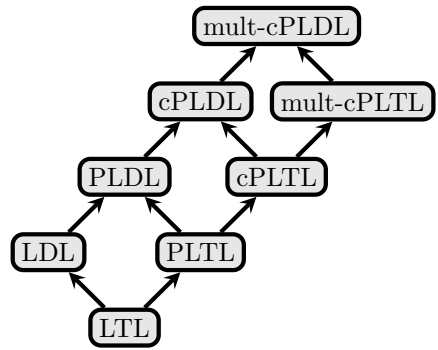

We introduced parameterized temporal logics whose operators bound the accumulated cost instead of time as usual: $\pltlc$ and $\pldlc$ extend $\pltl$ and $\pldl$ to the cost-setting while $\mpltlc$ and $\mpldlc$ extend them to the multi-dimensional cost-setting. The logics we considered here are depicted in Figure~\ref{fig_logics}: the upper four logics were introduced in this work. 

All four new logics retain the attractive algorithmic properties of $\ltl$ like a polynomial space model checking algorithm and a doubly-exponential time algorithm for solving infinite games. For $\pltlc$ and $\pldlc$, even the optimization variants of these problems are not harder than for $\pltl$: polynomial space for model checking and triply-exponential time for solving games, if the weights are encoded in unary. 

However, it is open whether these problems are strictly harder for logics with multiple cost functions or if the weights are encoded in binary. Another open question concerns the complexity of the optimization problem for infinite games: can these problems be solved in doubly-exponential time, i.e., is finding optimal variable valuations as hard as solving games? Note that this question is already open for $\pltl$. Recently, a step towards this direction was made by giving an approximation algorithm for this problem with doubly-exponential running time~\cite{TentrupWeinertZimmermann15}. Finally, one could consider weights from some arbitrary semiring and corresponding weighted parameterized temporal logics. 